\documentclass[twocolumn,showkeys,preprintnumbers,floatfix,amsmath,a4paper,nofootinbib,amssymb]{revtex4}

\usepackage[latin1]{inputenc}
\usepackage{graphicx}
\usepackage{amsthm}
\usepackage{courier}

\usepackage{graphicx}% Include figure files
\usepackage{bm}% bold math
\usepackage{color,soul}
\usepackage{amsmath,amsfonts,amssymb,amsthm}
\usepackage{mathtools}
\usepackage{commath}

\newtheorem{theorem}{Theorem}

\newtheorem{definition}{Definition}
\newtheorem{proposition}{Proposition}

\usepackage{xcolor}
\usepackage[colorlinks = true,
            linkcolor = blue,
            urlcolor  = blue,
            citecolor = blue,
            anchorcolor = blue]{hyperref}

\begin{document}

\title{The typical set and entropy in stochastic systems with arbitrary phase space growth}
\author{Rudolf Hanel$^{1,2}$ and Bernat Corominas-Murtra$^3$ }

\affiliation{
$^1$ Complexity Science Hub Vienna, Josefst\"adter Strasse 39, 1080 Vienna, Austria\\
$^2$ Section for Science of Complex Systems, Medical University of Vienna, Spitalgasse 23, 1090 Vienna, Austria\\
$^3$ Institute of Biology,  University Graz, Holteigasse 6, A-8010 Graz, Austria 
}

\thanks{author for correspondence: bernat.corominas-murtra@uni-graz.at}

\begin{abstract}
The existence of the {\em typical set} is key for data compression strategies and for the emergence of robust statistical observables in macroscopic physical systems. Standard approaches derive its existence from a restricted set of dynamical constraints. However, given the enormous consequences for the understanding of the system's dynamics, and its role underlying the presence of stable, almost deterministic statistical patterns, a question arises whether typical sets exist in much more general scenarios. We demonstrate here that the typical set can be defined and characterized  from general forms of entropy for a much wider class of stochastic processes than it was previously thought. This includes processes showing arbitrary path dependence, long range correlations or dynamic sampling spaces; suggesting that typicality is a generic property of stochastic processes, regardless of their complexity. Our results impact directly in the understanding of the stability of complex systems, open the door to new data compression strategies and points to the existence of statistical mechanics-like approaches to systems arbitrarily away from equilibrium with dynamic phase spaces. We argue that the potential emergence of robust properties in complex stochastic systems provided by the existence of typical sets has special relevance to biological systems.
\end{abstract}

\maketitle

%%%%%%%%%%%%%%%%%%%%%%%%%%%%%%%%%%%%%%%%%
%%
%%
%%		MAIN TEXT
%%
%%
%%%%%%%%%%%%%%%%%%%%%%%%%%%%%%%%%%%%%%%%%

Many natural systems are characterized by a high degree of internal stochasticity and for displaying processes leading to forms of organization of growing complexity \cite{Morowitz:1968, Maynard-Smith:1995,Bonner:1988, Wolpert:2007,Bialek:2012,Sole:2000, Tria:2014,Loreto:2016,Corominas-Murtra:2018,Iacopini:2020}. %Aside from the non-equilibrium processes of energy and entropy exchanges with the surroundings resulting from their open nature, the phase space volume {\em per particle} may not be considered constant, and conditions like the microscopic detailed balance may not be satisfied. 
Biological systems, at many scales, are paradigmatic examples of that, 
%\cite{Morowitz:1968, Maynard-Smith:1995,Bonner:1988, Wolpert:2007,Bialek:2012,Sole:2000}, 
triggering the debate whether the existence of {\em open-ended} evolution is a defining trait of them \cite{Schuster:1996, Bedau:2000, Kepa:2004, Kepa:2008, Day:2012, Packard:2019, Pattee:2019}, with the resulting challenge for a potential statistical-physics like characterization. In early embryo morphogenesis, for example, not only the number of cells increases exponentially in time, resulting into the corresponding increase of potential configurations, but also cells differentiate into specialized cell types \cite{Wolpert:2007}, implying, in statistical physics language, that new states enter the system. This process is almost completely irreversible and, although highly precise, is known to have a strong stochastic component \cite{Dietrich:2007,Maitre:2016,Giammona:2021}. On the other side, one can consider processes with collapsing phase spaces: Away from biology, recent advances in decay dynamics in nuclear physics succeeded considering a mathematical framework consisting on the stochastic collapse of the phase space \cite{Corominas-Murtra:2015, Corominas-Murtra:2017, Fujii:2021}. In figure (\ref{fig:CoinToss}) we schematically show the processes we are exploring. In spite of the ubiquity of such phenomena, a comprehensive characterization of systems with dynamic phase spaces in terms equivalent to the ensemble theory of statistical mechanics is lacking. 

Ensemble formalism in statistical mechanics is grounded on the concept of {\em typicality}  \cite{Cover:2012,Ash:2012,Pathria:2002, Pitowsky:2012, Lebowitz:1993}.  Informally speaking, given the set of all potential sequences of events resulting from a stochastic process, a subset, the {\em typical set}, carries most of the probability \cite{Ash:2012,Cover:2012}. This should not be confused with the set of most probable sequences: in the case of the biased coin, for example, the most probable sequence is not in the typical set. Instead, what it implies is that, for long enough sequences, the probability that the observed sequence or state belongs to the subset of sequences forming the typical set goes to $1$. Accordingly, a typical property for a stochastic system is robust and acts as a strong, almost deterministic attractor as long as the process unfolds \cite{Pitowsky:2012}, and one expects to observe it in the vast majority of cases. Moreover, if such a typical property exists, one can use this single property to --at least partially-- characterize the system, hence avoiding to go to the detailed, often unaffordable, microscopic description of all system's components. 
%A paradigmatic case is the thermodynamic temperature or entropy: A scalar parameter can be used to describe a system whose microscopic dynamics --in that case, kinetic energy or the logarithm total amount of degrees of freedom, respectively-- would require an astronomically high dimensional space to be characterized. 
Arguably, considerations based on typicality drive the connection between microscopic dynamics and macroscopic observables \cite{Lebowitz:1993, Battermann:2001,Frigg:2009}, and underlie the existence of the thermodynamic limit and, hence, the consistence between micro-canonical and canonical ensembles. 
In the context of information theory, the existence of the typical set for the outcomes of a given an information source has deep consequences in the process of data compression \cite{Cover:2012,Ash:2012}.
\begin{figure}[ht!]
\includegraphics[width=8.3cm]{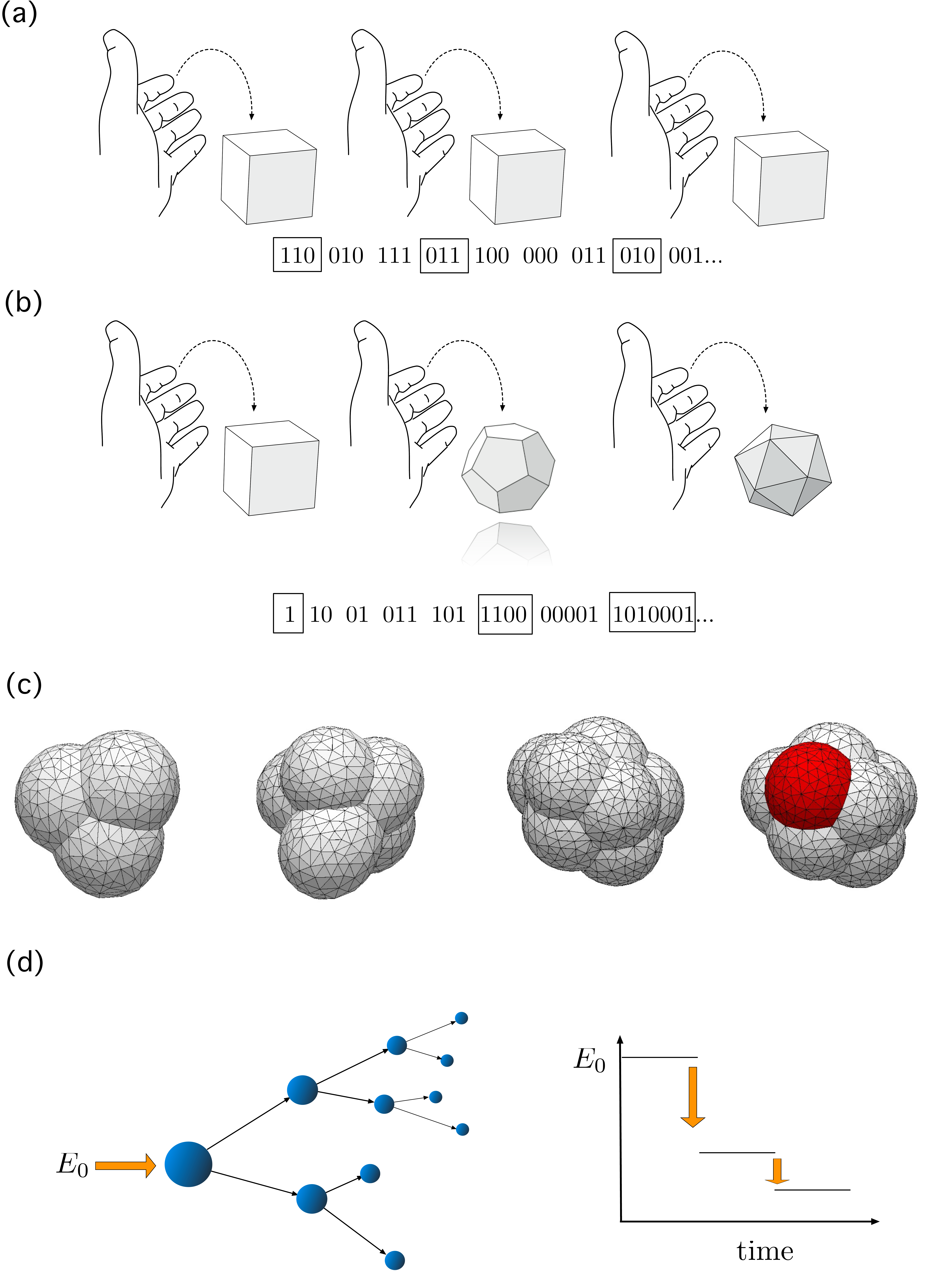}			%
 \caption{(a) Independent drawings of the same dice, either fair or biased, define a i.i.d. stochastic processes whose typical set is well defined and its growing  is approximately exponential \cite{Cover:2012}. (b) An example of a system whose typical set may show a super-exponential growth: At every drawing we update the dice by adding, e.g., a new face. (c) Potential configurations of early embryo development resembles, intuitively, the picture of the dice with growing faces. In this biological setting, new cells appear and, with that, new configurations but, on top of that, cells differentiate into new types --shown here in red-- adding new states in the system that were not there before. Interestingly, even highly reproducible, the whole process displays a strong stochastic component \cite{Dietrich:2007}. (d) Nuclear disintegration can be studied from the framework of collapsing phase spaces \cite{Fujii:2021}. In these processes, the amount of potential configurations of the system shrinks as long as the process unfolds. Toy models of embryo packings in (c) have been drawn using the \texttt{evolver} software package.}
\label{fig:CoinToss}	
\end{figure}

The size of the typical set gives us valuable information in relation to the particular way the stochastic process is filling the phase space. In equilibrium systems or information sources made of  independent drawings of identically distributed (i.i.d.) random variables, the Gibbs-Shannon entropic functional arises naturally in the characterization of the typical set \cite{Cover:2012, Ash:2012}, as the prefactor in the exponential describing the growth of its volume, establishing a clear connection between thermodynamics and phase space occupation. In systems/processes with collapsing or exploding phase spaces, path dependence or strong internal correlations \cite{Kac:1989, Pitman:2006, Clifford:2008, Tria:2014,Corominas-Murtra:2015,Loreto:2016,Corominas-Murtra:2018, Biro:2017, Jensen:2018, Iacopini:2020, Korbel:2021}, the phase space may grow super- or sub-exponentially, and the emergence of the Shannon-Gibbs entropic functional derived from phase space volume occupancy considerations is no longer guaranteed. The same situation may arise in cases dealing with non-stationary information sources  \cite{Gray:1974,Visweswariah:2000,Vu:2009,Boashash:2013,Granero:2019}. Generalized forms for entropies have been proposed to encompass these more general scenarios \cite{Abe:2000, Hanel:2011a, Enciso:2017, Tempesta:2016, Tempesta:2016a,Thurner:2017, Jizba:2020, Korbel:2019, Jizba:2017}, some of them explicitly linking the entropic functional to the expected evolution of the phase space volumes \cite{Hanel:2011, Tempesta:2016, Jensen:2018, Jensen:2018a,  Korbel:2018, Korbel:2020, Hanel:2013}.
In spite notable advances have been reported even for systems with physical significance \cite{Nicholson:2016,Balogh:2020, Korbel:2021}, the concept of typicality has not been yet explored for systems/processes with exploding or shrinking phase spaces, displaying path dependent dynamics or subject to internal correlations. 

The purpose of this paper is to fill this important gap in the theory of stochastic processes, providing results with potential implications in the theory of non-equilibrium systems and in data compression and coding strategies. As we shall see, the typical set can be defined for processes arbitrarily away from the i.i.d. frame, only assuming a very generic convergence criteria, satisfied by a broad class of stochastic processes, that here we refer to as {\em compact stochastic processes}.

\section{Results}
\subsection{Compact stochastic processes}
Let us consider a general class of  stochastic processes $\eta$ \cite{Gardiner:1983, Feller:1991}. We call this class {\em categorial processes} and they encompass almost any discrete stochastic process that can be conceived. A realization of $t$ steps of the process is denoted as $\eta(t)$:
\[
\eta(t)=\eta_1\eta_2. . .\eta_{t-1}\eta_t\quad,
\]
where $\eta_1, \eta_2. . .,\eta_{t-1},\eta_t$ are random variables themselves. Note that, in different realizations of $t$ steps of the process, the sequence of random variables can be different, as the process may display path dependence, long term correlations, or changes of the phase space, either shrinking or expanding. We denote a particular trajectory/path the process may follow as:  
\[
x(t)\equiv x_1x_2 . . .x_{t-1}x_t\in \Omega(t)\quad,%=\Omega_1\times . . .\Omega_t\quad,
\]
$\Omega(t)$ being the set of all possible paths of the process $\eta$ up to time $t$. We focus on the family of stochastic processes where there exists i) a positive, strictly concave and strictly increasing function $\Lambda\in\mathcal{C}^2$ in the interval $[1,\infty)$, such that $\Lambda(1)=0$, and ii) a positive, strictly increasing, $g\in\mathcal{C}^2$, in the interval $(1,\infty)$, by which:
\begin{eqnarray}
\lim_{t\to \infty}\frac{1}{g(t)}\Lambda\left(\frac{1}{p(\eta(t))}\right)=1\quad,
\label{eq:toinftyeta}
\end{eqnarray}
where the convergence is in probability \cite{Feller:1991}. We will call this family of stochastic processes {\em compact stochastic processes} (CSP). 
Given a CSP process $\eta$, a pair of functions $\Lambda,g$ by which equation (\ref{eq:toinftyeta}) is satisfied define a {\em compact scale} of the CSP process $\eta$. Note that these two functions may not be unique for a given process, meaning that the process can have several compact scales. 

It is straightforward to check that, if $\eta$ is a sequence of i.i.d. random variables $X_1, . . .,X_t\sim X$, $\Lambda=\log$ and $g(t)$ is $t$ times the Shannon entropy of a single realization, $H(X)$, the above condition holds, as it recovers the standard formulation of the Asymptotic Equipartition Property (AEP) \cite{Ash:2012, Cover:2012}. Therefore, the drawing of i.i.d. random variables $\sim X$ is a CSP with compact scale $(\log, H(X)t)$. However, the range of potential processes is, in principle, much broader. In consequence, the first question we ask concerns the constraints that the convergence condition (\ref{eq:toinftyeta}) imposes on $\Lambda$. Assuming that (\ref{eq:toinftyeta}) holds, one finds that $\Lambda$'s satisfying the following condition are candidates to characterize CSP's --see proposition \ref{Sprop:Lambdaz} of the Supplementary Information (SI) for details:
\begin{equation}
\lim_{z\to \infty}\frac{\Lambda(\lambda z)}{\Lambda(z)}=1\quad, \forall \lambda\in\mathbb{R}^+\quad.
\label{eq:lambdaz}
\end{equation}
Typical candidates for $\Lambda$ are of the form $\Lambda(z)=c\log^d(z)$, where $c,d$ are two positive, real valued constants or, more generally:
\[
\Lambda(z)=c_1\log^{d_1}(1+c_2\log^{d_2}(1+ c_3\log^{d_3}(. . . )))\quad,
\]
where $c_1, . . .$ and $d_1 , . . .$ are positive, real valued constants. In previous approaches, these constants have been identified as scaling exponents that enabled us to classify the different potential growing dynamics of the phase space \cite{Korbel:2018}. 

We observe that for CSP's, equation (\ref{eq:toinftyeta}) directly implies that there are two non-increasing sequences of positive numbers $\epsilon_1, . . .\epsilon_t, . . .$, $\delta_1, . . .\delta_t, . . .$, with $\lim_{t\to\infty}\epsilon_t=\lim_{t\to\infty}\delta_t=0$, from which there is a subset of paths $A[\epsilon_t]\subseteq\Omega(t)$ by which, for all $x(t)\in A[\epsilon_t]$:
\begin{equation}
\Lambda^{-1}((1+\epsilon_t)g(t))\leq p(x(t))\leq  \Lambda^{-1}((1-\epsilon_t)g(t))\quad,
\label{eq:TypicalProb}
\end{equation}
and:
\begin{equation}
\mathbb{P}(A[\epsilon_t])> 1-\delta_t\quad.
\label{eq:PA>1}
\end{equation}
where:
\[
\mathbb{P}(A[\epsilon_t])=\sum_{x(t)\in A[\epsilon_t]}p(x(t))\quad.
\]
We call the sequence of subsets $A[\epsilon_1]. . .A[\epsilon_t]$ of the respective sampling spaces $\Omega(1) . . .\Omega(t)$ a {\em sequence of typical sets of $\eta$}. Informally speaking, equation (\ref{eq:PA>1}) tells us that, for large enough $t$'s, the probability of observing a path that does not belong to the typical set becomes negligible. In consequence, {\em the typical set can be identified for CSP's}: Given CSP, the typical set $A[\epsilon_t]$ absorbs, in the limit $t\to\infty$, all the probability --see Theorem \ref{Stheorem:Typical} of the appendix.  We omitted a direct reference to the process $\eta$ in the notation of the typical set  (i.e.: $A[\epsilon_t]\equiv A[\epsilon_t](\eta)$) for the sake of readability. In the sequel we will omit this reference unless it is strictly necessary.
 In the next section we provide more details on the specific bounds in size by studying a subclass of the CCP's, namely, the class of {\em simple} CCP's. For them, the characterization of the typical set can be performed from a generalized form of entropy. 

\subsection{The typical set and generalized entropies}
%{\em The typical set and generalized entropies}.- 
Equation (\ref{eq:toinftyeta}) can be related to a general form of path entropy:
\begin{equation}
S_{\Lambda}(\eta(t))=\sum_{x(t)\in\Omega(t)} p(x(t))\Lambda\left(\frac{1}{p(x(t))}\right)dx(t)\quad,
\label{eq:genentrop}
\end{equation}
It can be proven that $S_{\Lambda}$ satisfies three of the four Shannon-Khinchin axioms expected by an entropic functional \cite{Shannon:1948, Khinchin:1957, Ash:2012} in Khinchin's formulation \cite{Khinchin:1957}, to be referred as SK1, SK2, SK3. In particular SK1 states that entropy must be a function of the probabilities, which is satisfied by $S_\Lambda$, by construction. SK2 states that $S_{\Lambda}$ is maximized by the uniform distribution $q$ over $\Omega(t)$, i.e.:
\[
q(x(t))=\frac{1}{|\Omega(t)|}\quad.
\]
We further observe that $S_\Lambda$ is a monotonously increasing function as well, in the case of uniform probabilities: Let us suppose two CSP's $\eta$ and $\eta'$ that sample uniformly their respective sampling spaces, $\Omega(t),\Omega'(t)$, such that $|\Omega(t)|<|\Omega'(t)|)$. Let, in consequence, $q$ and $q'$ be the uniform distributions over $\Omega(t)$ and $\Omega'(t)$, respectively, then:
\[
S_{\Lambda}(q)=\Lambda\left(|\Omega(t)|\right)<\Lambda \left( |\Omega'(t)| \right)=S_{\Lambda}(q')\quad,
\]
where $S_{\Lambda}(q), S_{\Lambda}(q')$ are the generalized entropies as defined in equation (\ref{eq:genentrop}) applied to distributions $q$ and $q'$.
Finally, SK3 states that, if $p(x(t))=0$, then $p(x(t))$ does not contribute to the entropy, which implies:
\[
\lim_{p(x(t))\to 0}
p(x(t))\Lambda\left(\frac{1}{p(x(t))}\right)=0\quad,
\]
satisfied as well for any $\Lambda$ considered in the definition of the CSP's. In the proposition \ref{Sprop:SKs} of the SI we provide details of the above derivations. We observe that SK4 is not generally satisfied: This axiom states that $S(AB)=S(A)+S(B|A)$, and one can only guarantee its validity in the case of Shannon entropy, where $\Lambda=\log$. In the general case, this condition may not be satisfied. A different arithmetic rule can substitute SK4 to accomodate other entropic forms \cite{Tempesta:2016}. Notice, however, that the use of Shannon (path) entropy --i.e., $\Lambda=\log$-- in the compact scale of a CSP may be used in a very general case, including systems with correlations or super-exponential sample space growth, as we will see in section \ref{sec:CRP}.

If the contributions to the above entropy of the paths belonging to the complementary set of $A[\epsilon_t]$, $\Omega(t)\setminus A[\epsilon_t]$ are negligible in the limit of $t\to\infty$, then we call the CSP {\em simple}.
In the case of simple CSP's, the convergence condition (\ref{eq:toinftyeta}) can be rewritten as:
\begin{eqnarray}
\lim_{t\to \infty}\left|\frac{1}{g(t)}\Lambda\left(\frac{1}{p(\eta(t))}\right)-\frac{S_{\Lambda}(\eta(t))}{g(t)}\right|=0\quad,
\label{eq:toinftyeta_0}
\end{eqnarray}
(in probability). In consequence:
\begin{equation}
\frac{S_{\Lambda}(\eta(t))}{g(t)}\to 1\quad.
\label{eq:SLto1}
\end{equation}
In the theorem \ref{STh:Simple} of the SI we demonstrate this general result. 
Once condition (\ref{eq:toinftyeta_0}) is satisfied, the typical set can be naturally defined for CSP's in terms of the generalized entropy $S_\Lambda$. We first reword condition (\ref{eq:toinftyeta_0}) as follows: Given a simple CSP $\eta$, there are two non-increasing sequences of positive numbers $\epsilon_1, . . .\epsilon_t, . . .$, $\delta_1, . . .\delta_t, . . .$, with $\lim_{t\to\infty}\epsilon_t=\lim_{t\to\infty}\delta_t=0$, by which:
\begin{equation}
\mathbb{P}\left(\left|\frac{1}{S_\Lambda(\eta(t))}\Lambda\left(\frac{1}{p(x(t))}\right)-1\right|>\epsilon_t \right)<\delta_t\quad.
\label{eq:convergence}
\end{equation}
If condition (\ref{eq:convergence}) applies, for each $t>0$ there is a set of paths, the {\em typical set} $A[\epsilon_t]\subseteq\Omega(t)$, defined as:
\begin{equation}
A[\epsilon_t]=\left\{x(t) \in\Omega(t):\left|\frac{1}{S_\Lambda(\eta(t))}\Lambda\left(\frac{1}{p(x(t))}\right)-1\right|<\epsilon_t\right\},
\label{eq:Typical_Set}
\end{equation}
by which
%\begin{equation}
$\mathbb{P}(A[\epsilon_t])>1-\delta_t$. Notice that, now, the characterization of the typical set is made using the generalized entropy $S_\Lambda$.
%\label{eq:P1delta}
%\end{equation}

The next obvious question refers to the cardinality of the typical set $|A[\epsilon_t]|$. We will see that it can be bounded by above and below in a way analogous to the standard one \cite{Cover:2012}. We can provide the first bound by observing that: 
\begin{eqnarray}
1-\epsilon_t&\leq&\sum_{x(t)\in A[\epsilon_t]} p(x(t))\nonumber\\
&\leq& \frac{|A[\epsilon_t]|}{\Lambda^{-1}((1-\epsilon_t)S_\Lambda(\eta(t)))}\nonumber\quad,
\end{eqnarray}
where $\Lambda^{-1}$ is the inverse function of $\Lambda$, i.e., $(\Lambda^{-1}\circ \Lambda)(z)=z$, which exists given the assumption that $\Lambda$ is a monotonously growing function made in the definition of CSP's.
From that, it follows that the cardinality of the typical set is bounded from below as:
\begin{equation}
|A[\epsilon_t]| \geq (1-\epsilon)\Lambda^{-1}((1-\epsilon)S_\Lambda(\eta(t)))\quad. 
\label{eq:bbelow}
\end{equation}
For the upper bound,  we observe that:
\begin{eqnarray}
1&\geq &\sum_{x(t)\in A[\epsilon_t]} p(x(t))\nonumber\\
&\geq& \frac{|A[\epsilon_t]|}{\Lambda^{-1}((1+\epsilon_t)S_\Lambda(\eta(t)))}\nonumber\quad,
\end{eqnarray}
leading to:
\begin{equation}
|A[\epsilon_t]|\leq \Lambda^{-1}((1+\epsilon)S_\Lambda(\eta(t)))\quad.
\label{eq:babove}
\end{equation}
Given the bounds provided in equations (\ref{eq:bbelow}) and (\ref{eq:babove}), one can (roughly) estimate the volume of the typical set as --see proposition \ref{SI:prop:simplefillingtovolumetric} of the SI for details:
\begin{equation}
|A[\epsilon_t]|\approx \Lambda^{-1}(S_\Lambda(\eta(t)))\quad.
\label{eq:Raw_TypS}
\end{equation}
The above equation gives us the opportunity of rewriting the entropy in a {\em Boltzmann-like} form. Identifying the cardinality of the typical set as the effective number of alternatives the system can achieve, one can write:
\[
S_\Lambda(\eta(t))\sim\Lambda(|A[\epsilon_t]|)\quad.
\]
%We observe that, given the concavity of $\Lambda$, a continuous function $f$ such that for all $t>0$ $f(t)=|A[\epsilon_t]|$ is strictly convex. Notice that the opposite would imply the paradoxical situation where new states enter into the phase space while reducing the size of the typical set. In addition, we observe that, by construction, equation (\ref{eq:Raw_TypS}) does not depend on the specific choices of the compact scale $\Lambda,g$, as soon as they satisfy the convergence criterion specified in equation (\ref{eq:toinftyeta}). 
Finally, we notice that we can (roughly) approximate the typical probabilities as:
\[
p(x(t))\approx \frac{1}{\Lambda^{-1}(S_\Lambda(\eta(t)))}\quad.
\]
We thus provided a general proof that the typical set exists and that it can be properly defined for a wide class of stochastic processes, the CSP's,  those satisfying convergence condition (\ref{eq:toinftyeta}). Moreover, we show that its volume can be bounded and fairly approximated as a function of the generalized entropy emerging from the convergence condition, $S_\Lambda$, as defined in equation (\ref{eq:genentrop}). 
%Now we explore the possibilities of the degeneracy we briefly mentioned while presenting the CSP's, namely, the fact that, given a CSP $\eta$, the pair of functions $\Lambda,g$ satisfying the convergence condition (\ref{eq:toinftyeta}) is not unique. 

\subsection{Example: A path dependent process} 
\label{sec:CRP}

%%%%%%%%%%%%%%%%%%%%%%%%%%%%%%%%%%%%%%%%%
%%
%%
%%		MODIFY & SIMPLIFY ACCORDING TO APPENDIX THEOREM
%%
%%
%%%%%%%%%%%%%%%%%%%%%%%%%%%%%%%%%%%%%%%%%

\begin{figure}
\centering
\includegraphics[width=8.cm]{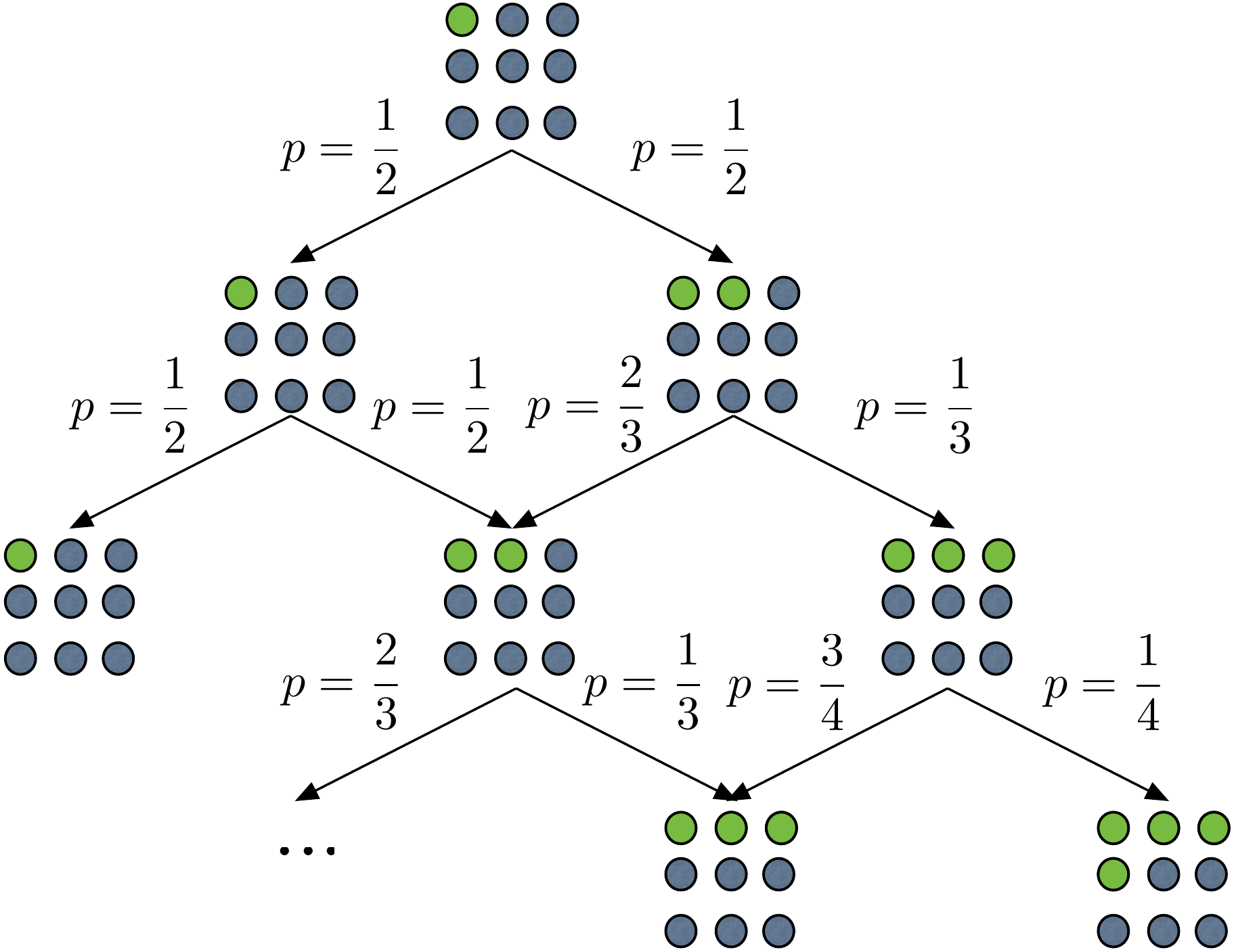}			%
 \caption{The rules of the Chinese restaurant process with memory. Here green circles represent occupied tables and grey circles empty tables --notice that in the mathematical formulation of the problem the number of tables is infinite. Arrows depict the possible transitions of the process and the associated probabilities. 
}
\label{fig:CRP}	
\end{figure}
\begin{figure*}[ht!]
\centering
\includegraphics[width=18.cm]{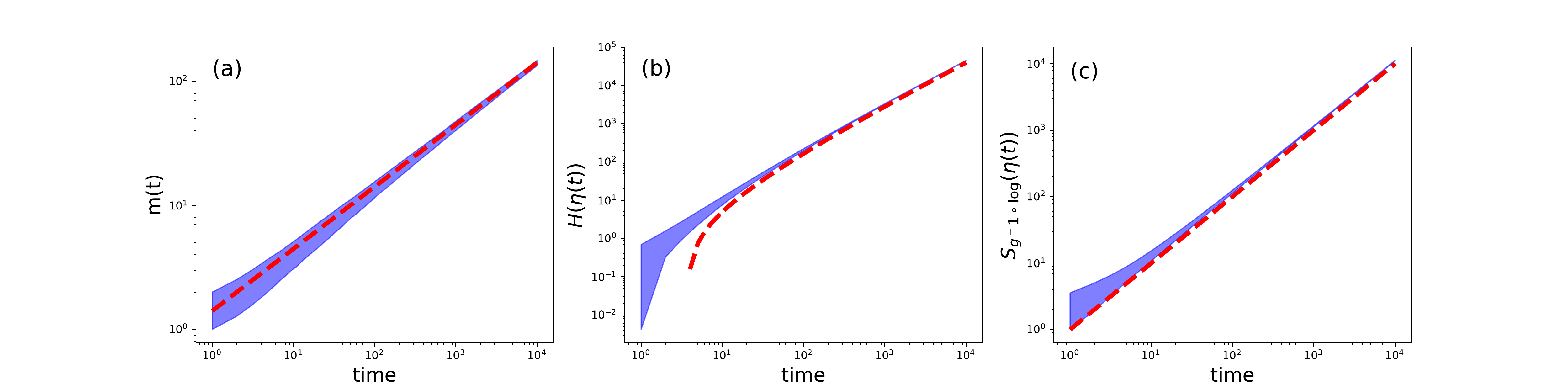}			%
 \caption{Numerical simulations for the {\em Chinese Restaurant process with memory}. The blue cloud represents actual numerical outcomes, dashed orange line the theoretical prediction. Time is given in arbitrary coordinates, representing a step in the process. In (a) we show the evolution of the amount of occupied tables against the prediction $m(t)\sim\sqrt{2t}$. (b) The evolution of Shannon path entropy for the CRPM, being the prediction given in (\ref{eq:phiCh}). The dashed red line shows the function $g(t)\sim \frac{t}{2}\log t$. (c) Evolution of the generalized path entropy $S_\Lambda$, with $\Lambda$ as defined in equation (\ref{eq:LambdaAlternative}). Numerical outcomes have been obtained from $1000$ replicas of the whole CRPM process up to $t=10^4$ steps.
}
\label{fig:Entropy_Plots}	
\end{figure*}

We briefly explore the behaviour of the typical set and its associated entropic forms through a model displaying both path dependence and unbounded growth of the phase space. The process $\eta$ works as follows: Let us suppose we have a restaurant with an infinite number of tables $m_1,. . . .,m_n, . . .$. At $t_0=0$ a customer enters the restaurant and sits at table $m_1$.
At time $t$ a new customer enters the restaurant where already $m(t)$ tables are occupied 
--occupation number of each table is unbounded. The customer can chose either sitting 
in an already occupied table from the $m_1, . . .,m_{m(t)}$ occupied tables, each with equal probability $\frac{1}{m(t) + 1}$, or in the next unoccupied one, $m_{m(t)+1}$, again with probability $\frac{1}{m(t)+1}$. This process is a version of the so-called {\em Chinese restaurant process} \cite{Pitman:2006, Bassetti:2009} with a minimal ingredient of memory/path dependence. Hence, we refer to it as the {Chinese restaurant process with memory} (CRPM). In figure (\ref{fig:CRP}) we sketch the rules of this process. Crucially, as $t\to\infty$, the random variable accounting for the number of tables $m(t)$ has the following convergent behaviour --see proposition \ref{prop:tosqrtmu} of the SI for details:
\[
\frac{m(t)}{\sqrt{2t}}\to 1\quad.
\]
In figure (\ref{fig:Entropy_Plots}a) we see that the prediction $m(t)\sim \sqrt{2t}$ is quite accurate when confronted to numerical simulations of the process. 
This property enables us to demonstrate that the CRPM we are studying is actually a CSP with compact scale $(\log, \frac{t}{2}\log t)$ --see theorem \ref{Sth:compactCrp} of the SI. In particular, equation (\ref{eq:toinftyeta}) is satisfied, in this particular case as:
\[
\lim_{t\to \infty}\frac{1}{\frac{t}{2}\log t}\log\left(\frac{1}{p(\eta(t))}\right)=1\quad,
\]
in probability. In addition, the process is {\em simple} --see theorem \ref{Sth:CRPSimple}. Since we are using $\Lambda=\log$, the entropy form that will arise is  Shannon path entropy, by direct application of equation \ref{eq:genentrop}, i.e., $S_\Lambda(\eta(t))=H(\eta(t))$, with $H(\eta(t))$ defined as: 
\begin{equation}
 H(\eta(t))=-\sum_{x(t)\in\Omega(t)}p(x(t))\log p(x(t))\quad.
\label{eq:phiCh}
\end{equation}
In consequence,
\begin{equation}
\frac{H(\eta(t))}{\frac{t}{2}\log t}\to 1\quad.
\label{eq:HT/t2}
\end{equation}
Given the compact scale used, one can estimate the evolution of the size of the typical set as:
\begin{equation}
|A[\epsilon_t]|\sim \sqrt{\Gamma(t)}\quad,
\label{AChinese}
\end{equation}
where $\Gamma$ is the standard $\Gamma$-function \cite{Abramowitz:1964}. We see that the growth of the typical set as shown in equation (\ref{AChinese}) is clearly faster than exponential. In addition,  in figure (\ref{fig:Entropy_Plots}b) we see that the prediction made in equation (\ref{eq:HT/t2}) fits perfectly with the numerical realizations of the process. Note that we have shown the dependence on Shannon path entropy for the clarity in the exposition. Indeed, as pointed out above, a CCP $\eta$ may have several compact scales. For example, taking the compact scale that led to Shannon entropy, $(\log, g(t))$, with $g(t)=\frac{t}{2}\log t$, one can construct another compact scale for the CRPM by composing $g^{-1}$ --which, by assumption, exists-- to both functions. In consequence, one will have a new compact scale $(\Lambda, \tilde{g})$, defined as:
\begin{equation}
\Lambda(t)=(g^{-1}\circ \log)(t)\sim \frac{2\log(t)}{\mathbf{W}(2\log(t))}\,,\quad\tilde{g}(t)=t\quad,
\label{eq:LambdaAlternative}
\end{equation}
being  $\mathbf{W}$ the {\em Lambert} function \cite{Abramowitz:1964}, where only the positive, real branch is taken into account. In figure  (\ref{fig:Entropy_Plots}c) we see that $S_\Lambda(\eta(t))$ fits perfectly $g(t)\sim t$, proving that $(\Lambda, t)$ is a compact scale for the CRPM --see also section 3\ref{secSI:Lg} of the SI. We observe that this particular compact scale makes the path entropy $S_\Lambda$ {\em extensive} when applied to the  CRPM.

\section{Discussion}  

%%%%%%%%%%%%%%%%%%%%%%%%%%%%%%%%%%%%%%%%%
%%
%%
%%		ADD STAT MECH OPEN ENDED
%%
%%
%%%%%%%%%%%%%%%%%%%%%%%%%%%%%%%%%%%%%%%%%

We demonstrated that, for a very general class of stochastic processes, to which we refer to as {\em compact stochastic processes}, the typical set is well defined. These processes can be path dependent, contain arbitrary internal correlations or display dynamic behaviour of the phase space, showing sub- or super- exponential growth on the effective number of configurations the system can achieve. The only requirement is that there exist two functions $\Lambda,g$ for which equation (\ref{eq:toinftyeta}) holds. Along the existence of the typical set, a generalized form of entropy naturally arises, from which, in turn, the cardinality of the typical set can be computed. 

The existence of the typical set in systems with arbitrary phase space growth opens the door to a proper characterization, in terms of statistical mechanics, of a number of processes, mainly biological, where the number of configurations and states changes over time. In particular, it paves the path towards the statistical-mechanics-like understanding of processes showing open-ended evolution. For example, this could encompass thermodynamic characterizations of --part of-- developmental paths in early stages of embryogenesis. The existence of the typical set, even in some extreme scenarios of stochasticity and phase space behaviour, may not be uniquely instrumental as a theoretical tool: As a speculative hypothesis, one may consider that typicality lays behind the astonishing reproducibility and precision of some biological processes. In this scenario, stochasticity would drive the system to the set of {\em correct} configurations --those belonging to the typical set-- with high accuracy. Selection, in turn, would operate on typical sets, thereby promoting certain stochastic processes over the others. More specific scenarios are nevertheless required in order to make this suggesting hypothesis more sound.

Further works should clarify the potential of the proposed probabilistic framework to accommodate generalized, consistent forms of thermodynamics and explore the complications that can arise due to the break of ergodicity that is implicit in some of the processes compatible with the above description. Importantly, our results provide a potential starting point for an ensemble formalism for systems with arbitrary phase space growth, extending the concept of thermodynamic limit to these systems without requiring further conditions like microscopic detailed balance. Questions like the definition of free energies or the possible need of extensivity to have a consistent picture remain, however, open. To give tentative answers to these questions, links to early proposals could be in principle drawn, both at the level of thermodynamic grounds --see, e.g., \cite{Abe:2001, Abe:2006, Jensen:2018}-- and at the level of entropy characterization, as, for example, in \cite{Jensen:2018a, Tempesta:2016, Korbel:2018, Hanel:2011a, Korbel:2021}. We finally point out the impact of our results for the study of information sources, given the important consequences the typical set has for optimal coding and data compression. The existence of the typical set in these broad class of information sources, where in general, roughly speaking, the information flow is not constant, may open the possibility of new compressing strategies. These could be based, for example, on the encoding of the specific CSP used to generate the information source and the $\Lambda, g$ functions used to ensure convergence. 

\section*{Acknowledgements}
The authors want to thank Petr Jizba and Artemy Kolchinsky for the helpful discussions that enabled us to improve the quality of the manuscript. B. C-M wants to acknowledge the helpful hints from Daniel R. Amor and the support of the field of excellence {\em Complexity in Life, Basic Research and Innovation} of the University of Graz.

%%%%%%%%%%%%%%%%%%%%%%%%%%%%%%%%
%%%%%%%%%%%%%%%%%%%%%%%%%%%%%%%%
%%%%%%%%%%%%%%%%%%%%%%%%%%%%%%%%
%%%%%%%%%%%%%%%%%%%%%%%%%%%%%%%%
%%%%%%%%%%%%%%%%%%%%%%%%%%%%%%%%
%%%%%%%%%%%%%%%%%%%%%%%%%%%%%%%%
%%%%%%%%%%%%%%%%%%%%%%%%%%%%%%%%
%%%%%%%%%%%%%%%%%%%%%%%%%%%%%%%%%%%%%%%%%%%%%%%%%%%%%%%%%%%%%%%%%%%%%%%%%%%%%%%%%%

%\newpage
%\bibliography{../../../Bib/AEP_Non_Exp}
%%% BILBLIOGRAPHY %%%%%%%%%%%%%%%%%%%%%%
%\end{thebibliography}
%\newpage
%\newpage
%%%%%%%%%%%%%%%%%%%%%%%%%%%%%%%%%%%%%%%%%%%%%%%%%%%%%%%
%%%%%%%%%%%%%%%%%%%%%%%%%%%%%%%%%%%%%%%%%%%%%%%%%%%%%%%
%%%%%%%%%%%%%%%%%%%%%%%%%%%%%%%%%%%%%%%%%%%%%%%%%%%%%%%
%%%%%%%%%%%%%%%%%%%%%%%%%%%%%%%%%%%%%%%%%%%%%%%%%%%%%%%
%%%%%%%%%%%%%%%%%%%%%%%%%%%%%%%%%%%%%%%%%%%%%%%%%%%%%%%
%%%%%%%%%%%%%%%%%%%%%%%%%%%%%%%%%%%%%%%%%%%%%%%%%%%%%%%
%%%%%%%%%%%%%%%%%%%%%%%%%%%%%%%%%%%%%%%%%%%%%%%%%%%%%%%

\appendix

\section*{Supplementary material}
In this Supplementary material we systematically develop the mathematical theory used in the main text of the manuscript "The typical set and entropy in stochastic systems with arbitrary phase space growth". The text is structured as follows: First, we define the class  of discrete stochastic systems that we call {\em categorial}, which contain almost anything that can be conceived. From them, we select the subclass of {\em compact} processes, namely, those satisfying the convergence condition stated in equation (\ref{eq:toinftyeta}) of the main text. For them, we prove the existence of a sequence of typical sets. Further, we define another subclass, the subclass of {\em simple} processes, namely, those by which the complement of the typical set has no finite contributions to a generalized entropy in the limit $t\to \infty$. In these processes, the sequence of typical sets can be defined in terms of the generalized entropy. Finally, we present an example of a path dependent process, and we show that is {\em compact} and {\em simple}. In consequence, we can compute the typical probabilities and the size of the typical sets, which is shown to grow super-exponentially in time. In figure (\ref{fig:Hierarchy}) --below-- we outline the hierarchy that our study induces over stochastic processes.
\begin{figure*}%[h!]
\includegraphics[width=14.9cm]{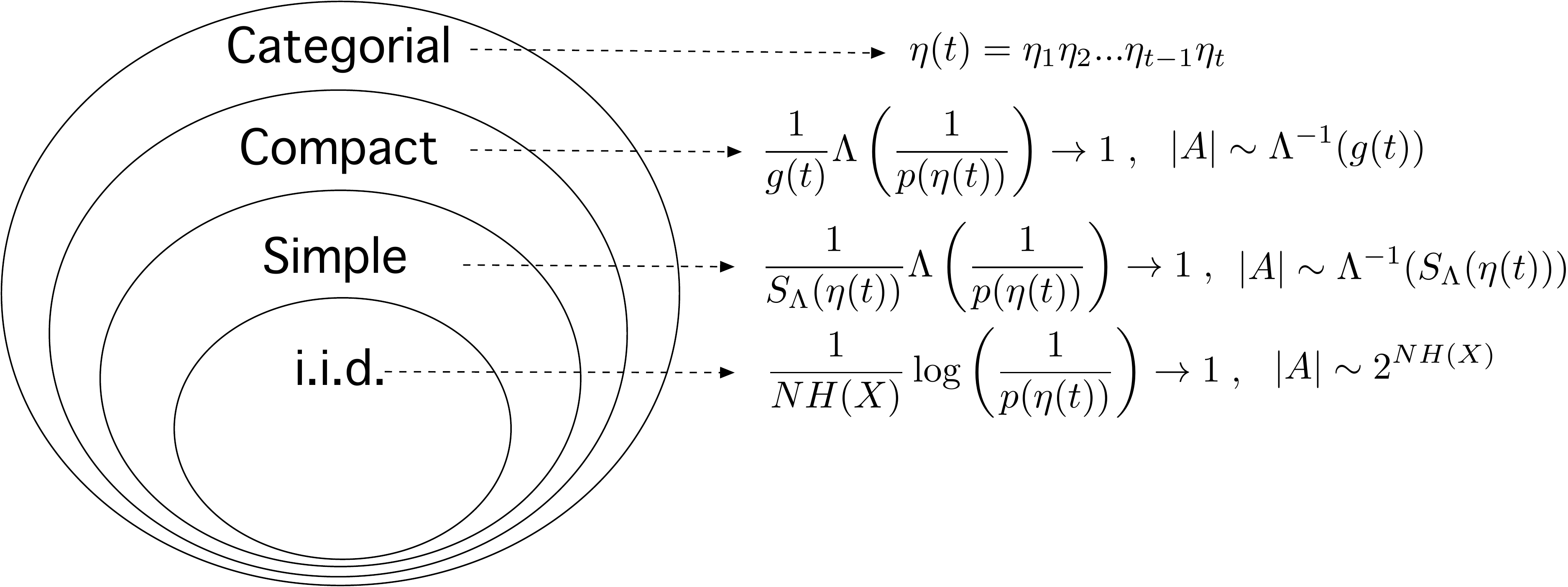}			%
 \caption{A potential hierarchy of discrete stochastic processes. The largest class would correspond to the {\em categorial} processes, which comprise almost anything that can be conceived as a discrete stochastic process. A subclass of categorial processes are the {\em compact} processes, by which the convergence condition stated in the equation (\ref{eq:toinftyeta}) of the main text holds and, therefore, a sequence of typical sets can be identified. Inside the compact processes, we identify the subclass of {\em simple} processes, where the sequence of typical sets can be defined in terms of a general form of entropy. Finally, the simplest subclass is the one defined by stochastic processes defined by sequences of independent, identically distributed random variables $\sim X$, by which the sequence of typical sets can be defined from the entropy in Shannon-like form. In turn, in i.i.d. systems, the path entropy up to time $t$ can be written as $t$ times the contribution of a single event \cite{Cover:2012}. Note that no assumptions of independence or stability of the sampling space are needed in the two first subclasses, even the typical set can be consistently identified. In addition, as we will see in section \ref{SSec:CRP}, the use of Shannon entropy to characterize the sequence of typical sets is not restricted to i.i.d. systems.}
\label{fig:Hierarchy}	
\end{figure*}

%%%%%%%%%%%%%%%%%%%%%%%%%%%%%%%%%%%%
%%%%%%%%%%%%%%%%%%%%%%%%%%%%%%%%%%%%
%%%%%%%%%%%%%%%%%%%%%%%%%%%%%%%%%%%%
%%%%%%%%%%%%%%%%%%%%%%%%%%%%%%%%%%%%
\section{Compact categorial processes and typical sets}
%%%%%%%%%%%%%%%%%%%%%%%%%%%%%%%%%%%%
%%%%%%%%%%%%%%%%%%%%%%%%%%%%%%%%%%%%
%%%%%%%%%%%%%%%%%%%%%%%%%%%%%%%%%%%%
%%%%%%%%%%%%%%%%%%%%%%%%%%%%%%%%%%%%

\subsection{Categorial processes}
$\\$

Categorial processes $\eta$ are processes that at any time $t$ 
 sample one state from a finite number of distinguishable states collected in the set $\Omega_t$, called the \textit{sample space} of the process at time $t$. If we look at a discrete time line $T={1,2,3,\cdots}$ we represent 
the process $\eta$ up to time $t$ as a sequence of random variables $\eta_t$, i.e.:
\[
\eta(t)=\eta_{1}\eta_{2}\cdots\eta_{t-1}\eta_t\, .
\]
The processes $\eta$ neither needs to consist of statistically independent random variables $\eta_t$ nor does the sample space of the process need to be constant. That is, the local sample spaces $\Omega_t$ of the variable $\eta_t$ can differ from the sample space $\Omega_{t'}$ of another variable $\eta_{t'}$. If we sample $\eta(t)$ this provides us with a particular path:
\[
x(t)=x_{1}x_{2}\cdots x_{t-1} x_t\, ,
\]
with $x(t)\in\Omega(t)$, being $\Omega(t)$ defined as:
\begin{equation}
\Omega(t)=\Omega_{1}\times\Omega_{2}\times \cdots \times\Omega_{t-1}\times\Omega_t\,,
\label{eq:dynamical_sample_spaces}
\end{equation}
in words, the Cartesian product of all local sample spaces $\Omega_{t'}$ up to time $t$.

In principle it could be that also the sample space $\Omega_t$ depends on the path $x(t-1)$ the process has taken
up to time $t-1$. In consequence,  $\Omega(t)$ could contain many paths that are not possible for the processes $\eta$ and, therefore, have zero probability  of being sampled. However, we consider processes by there exists a non-empty subset $\mathring\Omega(t) \subseteq \Omega(t)$, which contains all and only the sequences $x(t)$ of $\eta$ with $p(x(t))>0$.
\begin{definition}\label{def:wellformedinterior}
Let $\eta$ be a categorial process. We call and $p(x(t))$ the probability of paths $x(t)\in\Omega(t)$ to 
be sampled by $\eta$ and, then, we call:
\begin{equation}
\mathring\Omega(t)=\{x(t)\in\Omega(t)\,:\,p(x(t))>0\}\,,
\label{eq:interior_dynamical_sample_spaces}
\end{equation}
the \textbf{well formed interior} of $\Omega(t)$ or the 
\textbf{path sample space} of the process.
\end{definition}
Let us call $\Omega_t[x(t-1)]$ the set of states of $\Omega_t$ that can be sampled at time $t$ {\em provided that} our trajectory up to time $t-1$ was $x(t-1)$. The set of potential states that can be visited at time $t$, will then be: 
\[
\Omega_t=\bigcup_{x(t-1)\in\mathring\Omega(t-1)}\Omega_t[x(t-1)]\,,
\]
 i.e. $\Omega_t$ contains all states the process could possibly sample at time $t$ after all possible histories the process could have sampled up to time $t-1$. In this way we can always assume that we can find $\Omega(t)$ and its well formed interior $\mathring\Omega(t)$ by pruning all ill formed sequences from $\Omega_t\times \mathring\Omega(t-1)$. 
All information on how the sample space $\Omega_t$ gets sampled then solely resides in the hierarchy of transition probabilities $p(x_t|x(t-1))$, where $x_t\in\Omega_t$ and 
$x(t-1)\in \Omega(t-1)$. We get
\[
\mathring\Omega(t)=\{x(t)\in\Omega_t\times \mathring\Omega(t-1)\,:\,p(x_t|x(t-1))>0\}\,.
\]
In the context of categorial processes,  $p(x(t))$ is clearly a monotonic decreasing function in time bounded below by zero, $p(x(t))\geq 0$. In consequence, $\lim_{t\to\infty} p(x(t))$ converges for all paths possible paths $x(t)$ of the process $\eta$. This also means that for almost all paths the probability  $\lim_{t\to\infty} p(x(t))=0$ even though some finite number of paths could have non-zero probabilities even in the limit $t\to\infty$, although convergence is guaranteed\footnote{Unlike for processes, for {\em systems} --e.g., particles in a box-- it is not guaranteed that adding a new particle --as analog of to a new sample step-- $p(x(t))\leq p(x(t-1)$. In consequence, for systems, we cannot guarantee
convergence of $p(x(t))$ as the system size $t\to\infty$}.  After this general description of categorial processes, we can start by characterizing the subclass of them we are interested in. 

\subsection{Compact categorial processes}
$\\$

In the following we provide the condition of compactness that we impose to categorial processes in order to ensure the existence of a typical set.
\begin{definition}\label{def:compact}
Let $\eta$ be a categorial process $\eta$ and let $\eta(t)$ denote the process up to time $t$ and
$\Omega(t)$ denote the path sample space of $\eta(t)$ (as discussed above). Let us consider pairs of functions $(\Lambda,g)$, such that $\Lambda$ is a twice continuously differentiable, strictly monotonic increasing and strictly concave function on the interval $[1,\infty)$ with $\lim_{t\to\infty}\Lambda(t)=\infty$ and $\Lambda(1)=0$; and $g$ is a twice continuously differentiable strictly monotonically increasing function on the interval $[0,\infty)$. If we can \textbf{associate} such a pair of functions $(\Lambda,g)$ with the process $\eta$, such that:
\[
\lim_{t\to\infty}\frac{1}{g(t)}\Lambda\left(\frac{1}{p(\eta(t))}\right)\ =\ 1\,,
\]
(in probability), then we call the process $\eta$ \textbf{compact} in $(\Lambda,g)$ and $(\Lambda,g)$ a
\textbf{compact scale} of $\eta$.
\end{definition}
Note that compact scales $(\Lambda,g)$ associated to a given \textit{compact categorial processes} (CCP) need not be unique\footnote{The fact that we can find different pairs of functions, $(\Lambda,g)$, in which a process $\eta$ is compact, leads to questions related to equivalence relations on the space of pairs, $(\Lambda,g)$. As it turns out, following up the  idea of typical sets, that we are to explode below, a process $\eta$ induces an equivalence relation on this space, partitioning the space into monads of equivalent compact scales, $(\Lambda,g)$, of the process. At the same time this means that there exist inequivalent compact scales, which can be thought of as different "scales of resolution" to look at a process.}.

\subsection{Typical sets in compact categorial processes}
$\\$

Now that we have defined the stage for CCPs we can define particular sequence of subsets of paths that tell us where the probability of finding paths "typically" localizes in the path sample space. 

\begin{definition}
\label{def:compact}
Let $\eta$ be compact in $(\Lambda,g)$ and let $\epsilon_1, . . .,\epsilon_t, . . .$ be a non-increasing sequence with limit $\lim_{t\to\infty}\epsilon_t=0$, then we can define the set:
\[
A[\epsilon_t]=\left\{x(t)\in\Omega(t)\,:\, \left|\frac{1}{g(t)}\Lambda\left(\frac{1}{p(\eta(t))}\right)-1 \right|< \epsilon_t\right\}\,,
\]
for every $t=1,2,3,\cdots$. If for the choice of $\epsilon_t$ it holds that 
\[
\lim_{t\to\infty}\mathbb{P}(x(t)\in A[\epsilon_t])=1\,,
\]
then we call $A[\epsilon_t]$ a \textbf{typical set} in $(\Lambda,g)$ at time $t$ and the sequence $\epsilon_1, \epsilon_2, . . .,\epsilon_{t-1},\epsilon_t$ 
a \textbf{typical localizer} of $\eta$.
\end{definition}
At this point we have introduced the notion of typicality. Next we proof the following:
\begin{theorem}
\label{Stheorem:Typical}
If $\eta$ is a CCP in $(\Lambda,g)$, then there exist (a) typical localizer sequence $\epsilon_1, \epsilon_2, . . .,\epsilon_{t-1},\epsilon_t, . . .$  associated to $\eta$ and hence (b) the respective sequence of typical sets $A[\epsilon_1], . . .,A[\epsilon_t], . . .$ by which:
\[
\lim_{t\to \infty}\mathbb{P}(x(t)\in A[\epsilon_t])=1\,.
\]
\end{theorem}

\begin{proof}
Since $\eta$ is a CCP we know, by assumption, that, as $t\to\infty$ $\Lambda(1/p(\eta(t)))/g(t)\to 1$ (in probability). We can rewrite this condition by stating that, for every $\epsilon, \delta>0$, there exists a $t_0$ by which, for each $t>t_0$ \cite{Feller:1991}:
\[
\mathbb{P}\left(\left|\frac{1}{g(t)}\Lambda\left(\frac{1}{p(x(t))}\right)-1\right|>\epsilon\right)<\delta\,.
\]
Let $\tau(\varepsilon,\delta)$ be the smallest such $t_0$. We can then use two arbitrary strictly monotonic decreasing functions $\epsilon^*_n$ and $\delta^*_n$ that converge to zero and construct a monotonically increasing sequence of times $t_n=\tau(\epsilon^*_n,\delta^*_n)$ such that for all $t\geq t_n$ it is true that
\[ 
\mathbb{P}\left(\left|\frac{1}{g(t)}\Lambda\left(\frac{1}{p(x(t))}\right)-1\right|>\epsilon^*_n\right)<\delta^*_n\,.
\]
From that, it is straightforward to define a typical localizer sequence $\epsilon_1, . . .,\epsilon_t, . . .$, by just taking:
\[
(\forall t: t_n\leq t<t_{n+1})\,,\,\epsilon_t=\epsilon^*_n\,.
\]
Finally, the condition:
\[
\lim_{t\to \infty}\mathbb{P}(x(t)\in A[\epsilon_t])=1\,,
\]
follows as a direct consequence of the construction of the sequence of typical sets $A[\epsilon_1], . . .,A[\epsilon_t], . . .$, thereby concluding the proof.
\end{proof}

The next step is to check what kind of functions $\Lambda$ can be expected when dealing with CCP's. To that end, we will impose a condition to the CCP, namely, that the CCP is {\em filling}. From this --very mild-- condition, we will then check which functions enable the convergence criteria to be fullfilled. First of all, we need to introduce some technical terms.
\begin{definition}\label{def:typical_ratios}
Let $\eta$ be a CCP in $(\Lambda,g)$ with a typical localizer sequence $\epsilon_1, . . . , \epsilon_t, . . .$. We define the upper and lower \textbf{typical ratios}, $r_{\pm}(t)$,  in $(\Lambda,g)$ as:
\[
r_{\pm}(t)=\frac{\Lambda^{-1}\left( (1\pm \epsilon_t) g(t) \right)}{\Lambda^{-1}\left( g(t) \right)}\,,
\]
where $\Lambda^{-1}$ is the inverse function of $\Lambda$, which exists due to the strict monotonicity of $\Lambda$.
\end{definition}
Note that, by construction, $r_+(t)\geq 1$ and $r_-(t)\leq 1$.
\begin{definition}\label{def:filling}
We call a CCP $\eta$ \textbf{filling} in $(\Lambda,g)$ if its typical ratios have the property
$\lim_{t\to\infty}r_+(t)=\infty$ and $\lim_{t\to\infty}r_-(t)=0$.
\end{definition}

We can now say something about the shape of functions $\Lambda$.
\begin{proposition}
\label{Sprop:Lambdaz}
If $\eta$ is a filling CCP in $(\Lambda,g)$, then it follows that:
\[
\lim_{z\to\infty} \frac{\Lambda(\lambda z)}{\Lambda(z)}=1\,,
\]
for all $\lambda>0$.
\end{proposition}

\begin{proof}
We note that for filling CCP $\eta$ it is true that:
\[
\frac{1}{g(t)}\Lambda(r_{\pm}(t)\Lambda^{-1}(g(t)))\to 1\,,
\]
as $t\to\infty$.
We can rewrite $z_t=\Lambda^{-1}(g(t))$ and for $\lambda>1$ we find a $t_0$ such that $\lambda=r_+(t_0)$.
Therefore for all $t>t_0$ we find:
\[
1\leq \frac{\Lambda(\lambda z_t)}{\Lambda(z_t)}\leq \frac{\Lambda(r_+(t) z_t)}{\Lambda(z_t)} \to 1\,.
\]
The other case, $0<\lambda<1$, we prove analogously using $r_-$ instead of $r_+$, and the proposition follows. 
\end{proof}

To get an idea which kind of functions satisfy this condition we can look at the following example: 
\begin{proposition}
For any $c>0$ the function $\Lambda(z)=\log(z)^c$ has the property 
$\lim_{z\to\infty} \Lambda(\lambda z)/\Lambda(z)=1$ for all $\lambda>0$.
\end{proposition}

\begin{proof}
We can proof this by direct computation, i.e.:
\begin{eqnarray}
\frac{\Lambda(\lambda z)}{\Lambda(z)}&=&\left(\frac{\log(\lambda z)}{\log(z)}\right)^c \nonumber\\
&=&\left(\frac{\log(\lambda)+\log(z)}{\log(z)}\right)^c \nonumber\\
&=&\left(1+\frac{\log(\lambda)}{\log(z)}\right)^c\,.\nonumber
\end{eqnarray}

Since $\log(z)\to\infty$, one can easily read from the last line that the example family of functions $\Lambda$
fulfils the proposition.
\end{proof}

In general we can say that candidates for $\Lambda$ of filling CCPs are of the form $\Lambda(z)=c\log(z)^d$ for some
positive constants $c$ and $d$ or even slower growing functions of the form:
\[
\Lambda=c_1\log(1+c_2\log(1+c_3\log(\cdots)^{d_3})^{d_2})^{d_1}\,.
\]

%%%%%%%%%%%%%%%%%%%%%%%%%%%%%%%%%%%%%%%%%%
%%%%%%%%%%%%%%%%%%%%%%%%%%%%%%%%%%%%%%%%%%
%%%%%%%%%%%%%%%%%%%%%%%%%%%%%%%%%%%%%%%%%%
%%%%%%%%%%%%%%%%%%%%%%%%%%%%%%%%%%%%%%%%%%
%%%%%%%%%%%%%%%%%%%%%%%%%%%%%%%%%%%%%%%%%%
%%%%%%%%%%%%%%%%%%%%%%%%%%%%%%%%%%%%%%%%%%
%%%%%%%%%%%%%%%%%%%%%%%%%%%%%%%%%%%%%%%%%%
%%%%%%%%%%%%%%%%%%%%%%%%%%%%%%%%%%%%%%%%%%
%%%%%%%%%%%%%%%%%%%%%%%%%%%%%%%%%%%%%%%%%%
%%%%%%%%%%%%%%%%%%%%%%%%%%%%%%%%%%%%%%%%%%
%%%%%%%%%%%%%%%%%%%%%%%%%%%%%%%%%%%%%%%%%%
%%%%%%%%%%%%%%%%%%%%%%%%%%%%%%%%%%%%%%%%%%
%%%%%%%%%%%%%%%%%%%%%%%%%%%%%%%%%%%%%%%%%%

%%%%%%%%%%%%%%%%%%%%%%%%%%%%%%%%%%%%
%%%%%%%%%%%%%%%%%%%%%%%%%%%%%%%%%%%%
%%%%%%%%%%%%%%%%%%%%%%%%%%%%%%%%%%%%
%%%%%%%%%%%%%%%%%%%%%%%%%%%%%%%%%%%%
\section{Typical sets, simplicity condition and generalized entropies}
%%%%%%%%%%%%%%%%%%%%%%%%%%%%%%%%%%%%
%%%%%%%%%%%%%%%%%%%%%%%%%%%%%%%%%%%%
%%%%%%%%%%%%%%%%%%%%%%%%%%%%%%%%%%%%
%%%%%%%%%%%%%%%%%%%%%%%%%%%%%%%%%%%%

\subsection{Generalized entropies associated to CCP's}
$\\$

We start by defining the generalized entropy associated to a CCP $\eta$ in $(\Lambda,g)$

\begin{definition}\label{def:gents}
Let $\eta$ be a CCP in $(\Lambda,g)$, then we call the measure, $S_\Lambda$ of $\eta(t)$, defined as:
\begin{equation}
S_\Lambda(\eta(t))=\sum_{x(t)\in \Omega(t)} p(x(t))\Lambda\left(\frac{1}{p(x(t))}\right)\,,
\label{Seq:entropy}
\end{equation}
a \textbf{generalized path entropy} associated with $\eta$. Note that for a set $B\subset \Omega(t)$ the generalized entropy measure, $S_\Lambda(B)$, is given by $S_\Lambda(B)=\sum_{x(t)\in B} p(x(t))\Lambda\left(1/p(x(t))\right)$.
\end{definition}
In the following proposition we see that the above defined entropy satisfies three of the four Shannon-Khinchin's axioms for an entropy measure \cite{Khinchin:1957} ({\bf SK1}, {\bf SK2}, {\bf SK3}). The fourth axiom ({\bf SK4}) is not generally satisfied.
\begin{proposition}
\label{Sprop:SKs}
The entropy functional $S_\Lambda$defined in equation (\ref{Seq:entropy}) satisfies the first three of the four Shannon-Khinchin's axioms for an entropy measure as formulated in \cite{Khinchin:1957}:

\noindent
{\bf SK1} $S_\Lambda$ is a contiunous function only depending on the probabilities $p(x(t))$.

\noindent
{\bf SK2} $S_\Lambda$ is maximized if $(\forall p(x(t)))$
$p(x(t))=\frac{1}{|\Omega(t)|}$, i.e., equiprobability.

\noindent
{\bf SK3} If $p(x(t))=0$, then: $p(x(t))\Lambda\left(\frac{1}{p(x(t))}\right)=0$, 
i.e., events with zero probability have no contribution to the entropy.
\end{proposition}

\begin{proof}
$\\$

\noindent
To demonstrate {\bf SK1}, it is enough to observe that $S_\Lambda$ is only function of the probabilities and to take into account that, by assumption, $\Lambda\in{\cal C}^2$, therefore, $S_\Lambda$ continuous.
$\\$

\noindent
To demonstrate {\bf SK2}, we need to maximize the functional $\psi$, defined as:
\[
\psi=S_\Lambda(\eta(t))-\alpha\left (\sum_{x(t)\in \Omega(t)}p(x(t))-1\right)\,,
\] 
where $\alpha$ is a Lagrangian multiplier implementing the normalization constraint.
Maximizing $\psi$ with respect to a $p(x(t))$yields:
\[
0=\frac{\partial \psi}{\partial p(x(t))}=\Lambda(z)-z\Lambda'(z)-\alpha\,,
\]
where $z=1/p(x(t))$ and $\Lambda'$ is the first derivative of $\Lambda$.
Note that if the equation does not depend explicitly on $x(t)$ and if it has a unique solution then the proposition is proved, since all $p(x(t))$ have the same value. To see that a unique solution exists we need to show that 
$f(z)=\Lambda(z)-z\Lambda'(z)$ is strictly monotonic. To see that, it is enough to note that the first derivative of $f$ is given by $f'(z)=-z\Lambda''(z)>0$, since, by definition of compactness, $\Lambda\in {\cal C}^2$ and strictly concave; and therefore $\Lambda''(z)<0$.
$\\$

\noindent
Finally, to demonstrate that $S_\Lambda$ satisfies {\bf SK3} we apply the l'Hopital rule. First, by defining $y=1/z$, one has that: 
\[
\lim_{z\to 0}z\Lambda\left(\frac{1}{z}\right)=\lim_{y\to \infty}\frac{1}{y}\Lambda(y)\,.
\]
Then, considering that, by definition, $\Lambda\in {\cal C}^2$ is a strictly growing and concave function, one is led, after application to the l'Hopital rule for the limit, to:
\[
\lim_{y\to\infty}\frac{1}{y}\Lambda(y)=\lim_{y\to\infty}\Lambda'(y)=0\,,
\]
thereby concluding the proof.
\end{proof}
We observe that {\bf SK3} enables us to safely perform the sum for the entropy over the whole set of paths $\Omega(t)$, since $\sum_{x(t)\in \mathring\Omega(t)}(...)=\sum_{x(t)\in \Omega(t)}(...)$. 

\subsection{Simple CCP's}
$\\$

Now we define the condition of {\em simplicity}, namely, the property of processes by which the contributions to the entropy from paths outside the typical set vanish\footnote{It is however conceivable that CCPs exist that are not simple and the entropy in the limit has singular contributions from complement of the typical set. Even in probabilistic terms the complement of the typical set has measure zero, the presence of a large amount of highly improbable paths could give a non vanishing contribution to the entropy.}. 
\begin{definition}\label{def:intrinsicsimplicity}
Let $\eta$ be a filling CCP in $(\Lambda,g)$ with a typical localizer sequence $\epsilon_1, . . .,\epsilon_t, . . .$. Let $ A^c[\epsilon_t]=\Omega(t)\setminus A[\epsilon_t]$ be the complement of the typical set
in the well formed interior of $\Omega(t)$. We call $\eta$ \textbf{simple} if:
\[
\lim_{t\to \infty}\frac{1}{g(t)}S_\Lambda( A^c[\epsilon_t])=0\,.
\]
\end{definition}
Given a filling CCP in $(\Lambda,g)$  $\eta$, one can check if the simplicity condition is satisfied as follows: Let $p^*(x(t))$ be such that
\[
p^*(x(t))=\min_{\mathring\Omega(t)}\{p(x(t))\}\,,
\]
--recall that we select among those paths $x(t)\in \mathring\Omega(t)$ i.e., those by which $p(x(t))>0$. Then, assume the extreme case by which $(\forall x(t)\in  A^c[\epsilon_t], p(x(t))=p^*(x(t))$, thereby maximizing $\Lambda(1/p(x(t)))$ and, in consequence, the contribution of the complementary of the typical set to the entropy. If the following limit holds:
\[
\lim_{t\to\infty}\frac{\delta_t}{g(t)}\Lambda\left(\frac{1}{p^*(x(t))}\right)=0\,,
\]
then the process is simple. Note that the converse may not be true, there can be processes by which this proof does not hold but still, they are simple. For that, one must explore other strategies. 

\subsection{Generalized entropies and the typical set}
$\\$

We go now to the next step in the characterization of the typical set: For processes satisfying the simplicity condition, the typical set can be defined in terms of the generalized entropy $S_\Lambda$. This is consequence of the following theorem:
\begin{theorem}
\label{STh:Simple}
If $\eta$ is a simple CCP in $(\Lambda,g)$ with typical localizer sequence $\epsilon_1, . . ., \epsilon_t, . . .$ and corresponding sequence of typical sets $A[\epsilon_1], . . ., A[\epsilon_t], . . .]$, then:
\[
\lim_{t\to\infty}\frac{S_\Lambda(\eta(t))}{g(t)} =\lim_{t\to\infty}\frac{S_\Lambda(A[\epsilon_t])}{g(t)} =1\,.
\]
\end{theorem}

\begin{proof}
We will start with the second equality, namely:
\[
\lim_{t\to\infty}\frac{S_\Lambda(A[\epsilon_t])}{g(t)} =1\,.
\]
From the definition of typical sets we know that, for paths $x(t)\in A[\epsilon_t]$, it is true that:
\[
\frac{1}{\Lambda^{-1}((1+\epsilon_t)g(t))}\ \leq\ p(x(t))\ \leq\ \frac{1}{\Lambda^{-1}((1-\epsilon_t)g(t))}\,.
\]
In consequence, given that $\mathbb{P}(x(t)\in A[\epsilon_t]))>1-\delta_t$, one can bound $S_\Lambda(A[\epsilon_t])$ as:
\[
(1-\delta_t) (1-\epsilon_t)< \frac{S_\Lambda(A[\epsilon_t])}{g(t)}< (1-\delta_t) (1+\epsilon_t)\,.
\]
Since, by construction $\lim_{t\to\infty}\epsilon_t=\lim_{t\to\infty}\delta_t=0$, this second part of the theorem is proven.
From that, the statement of the theorem:
\[
\lim_{t\to\infty}\frac{S_\Lambda(\eta(t))}{g(t)} =\lim_{t\to\infty}\frac{S_\Lambda(A[\epsilon_t])}{g(t)} =1\,,
\]
follows directly given the assumption of simplicity.
\end{proof}
Therefore, as we advanced above, in simple  CCP's, one can define the typical set $A[\epsilon_t]$ as:
\[
A[\epsilon_t]=\left\{x(t)\in\Omega(t)\,:\, \left|\frac{1}{S_{\Lambda}(\eta(t))}\Lambda\left(\frac{1}{p(x(t))}\right)-1 \right|< \epsilon_t\right\}\,,
\]
and, consequently, the typical probabilities will be bounded by:
\begin{eqnarray}
(\forall x(t)\in A[\epsilon_t] )\;;\quad\quad\quad&&\nonumber\\
\,\frac{1}{\Lambda^{-1}(S_\Lambda(\eta(t))(1+\epsilon_t))}&<&p(x(t)) \nonumber\\
&<&\frac{1}{\Lambda^{-1}(S_\Lambda(\eta(t))(1-\epsilon_t))}\,.\nonumber
\end{eqnarray}
Finally, as shown in equations (\ref{eq:bbelow}) and (\ref{eq:babove}) of the main text, the cardinality of the typical set can be bounded, in terms of the generalized entropy $S_\Lambda$ as:
\begin{eqnarray}
(1-\epsilon_t)\Lambda^{-1}((1-\epsilon_t)S_\Lambda(\eta(t)))&<&|A[\epsilon_t]|\nonumber\\
&<& \Lambda^{-1}((1+\epsilon_t)S_\Lambda(\eta(t)))\,.\nonumber\\
\label{Seq:BoundS}
\end{eqnarray}
As a consequence of the above chain of inequalities, one can go further in the characterization of the generalized entropy and its relation to the typical set. 
Indeed, for simple CCPs,
$g(t)\sim S_\Lambda(\eta(t)\sim \Lambda(|A[\epsilon_t]|)$. We demonstrate that in the following proposition:
\begin{proposition}
\label{SI:prop:simplefillingtovolumetric}
Let $\eta$ be a simple CCP in $(g,\Lambda)$ with some typical localizer sequence $\epsilon_1, . . .,\epsilon_t, . . .$, then:
\[
\lim_{t\to\infty}\frac{\Lambda(|A[\epsilon_t]|)}{S_\Lambda(\eta(t))}\ =\ 1\,.
\]
\end{proposition}

\begin{proof}
from equation (\ref{Seq:BoundS}), one can derive the following chain of inequalities:
\begin{eqnarray}
\Lambda((1-\epsilon_t)\Lambda^{-1}((1-\epsilon_t)S_\Lambda(\eta(t))))&<&\Lambda(|A[\epsilon_t]|)\nonumber\\
&<& (1+\epsilon_t)S_\Lambda(\eta(t))\,.\nonumber
\end{eqnarray}
The last term has no difficulties. To explore the behaviour of the first one, just rename the term:
\[
z\equiv \Lambda^{-1}((1-\epsilon_t)S_\Lambda(\eta(t)))\,,
\]
and rewrite the first term of the inequality as:
\[
\Lambda((1-\epsilon_t)\Lambda^{-1}((1-\epsilon_t)S_\Lambda(\eta(t))))=\Lambda((1-\epsilon_t)z)\,.
\]
We know, from proposition \ref{Sprop:Lambdaz}, that the functions we are dealing with behave such that:
\[
\lim_{z\to\infty}\frac{\Lambda (\lambda z)}{\Lambda(z)} \to 1\,,\, (\forall z>0).
\]
As a consequence:
\[
\frac{\Lambda((1-\epsilon_t)\Lambda^{-1}((1-\epsilon_t)S_\Lambda(\eta(t))))}{S_\Lambda(\eta(t))}\to 1\,.
\]
Therefore, since also the third term goes trivially to $\sim S_\Lambda(\eta(t))$, we can conclude that:
\[
\frac{\Lambda(|A[\epsilon_t]|)}{S_\Lambda(\eta(t))}\ \to1\,,
\]
as we wanted to demonstrate.
\end{proof}

%%%%%%%%%%%%%%%%%%%%%%%%%%%%%%%%%%%%%%%%%%
%%%%%%%%%%%%%%%%%%%%%%%%%%%%%%%%%%%%%%%%%%
%%%%%%%%%%%%%%%%%%%%%%%%%%%%%%%%%%%%%%%%%%
%%%%%%%%%%%%%%%%%%%%%%%%%%%%%%%%%%%%%%%%%%
%%%%%%%%%%%%%%%%%%%%%%%%%%%%%%%%%%%%%%%%%%
%%%%%%%%%%%%%%%%%%%%%%%%%%%%%%%%%%%%%%%%%%
%%%%%%%%%%%%%%%%%%%%%%%%%%%%%%%%%%%%%%%%%%
%%%%%%%%%%%%%%%%%%%%%%%%%%%%%%%%%%%%%%%%%%
%%%%%%%%%%%%%%%%%%%%%%%%%%%%%%%%%%%%%%%%%%
%%%%%%%%%%%%%%%%%%%%%%%%%%%%%%%%%%%%%%%%%%
%%%%%%%%%%%%%%%%%%%%%%%%%%%%%%%%%%%%%%%%%%
%%%%%%%%%%%%%%%%%%%%%%%%%%%%%%%%%%%%%%%%%%
%%%%%%%%%%%%%%%%%%%%%%%%%%%%%%%%%%%%%%%%%%
%%%%%%%%%%%%%%%%%%%%%%%%%%%%%%%%%%%%%%%%%%
%%%%%%%%%%%%%%%%%%%%%%%%%%%%%%%%%%%%%%%%%%
%%%%%%%%%%%%%%%%%%%%%%%%%%%%%%%%%%%%%%%%%%

%\newpage

\section{The Chinese Restaurant process}
\label{SSec:CRP}

We now turn to analysing the version of the Chinese restaurant process with memory (CRPM) discussed in the main body of the paper. The version presented here is a variation of the standard Chinese Restaurant process as found in \cite{Pitman:2006, Bassetti:2009}.
\subsection{Definition and basics}
Suppose a restaurant with an infinite set of tables $m_1, . . .,m_n, . . .$ each with infinite capacity. The first customer enters and sits at the first table $m_1$. The second customer now has a choice to also sit down at the first table $m_1$ together with the first customer or to choose a free table $m_2$, each with probability $1/2$.  Let $m(t-1)$ be the number of occupied tables at $t-1$. If the $t$'th customer finds that $m(t-1)$ tables are already occupied by some guests, then again the customer will choose one of the non-empty tables:  
\[
m_1, . . .,m_{m(t-1)}\,,
\]
each with probability $1/(m(t-1)+1)$, in which case $m(t)=m(t-1)$, or the next empty table $m_{m(t-1)+1}$, also with probability 
$1/(m(t-1)+1)$. In this later case $m(t)=m(t-1)+1$. The key point is therefore the number of occupied tables $m(t)$. We observe that the amount of occupied tables it can be rewritten as a stochastic recurrence:
\begin{equation}
m(t+1)=m(t)+\zeta(m(t))\,,
\label{Seq:m(t)}
\end{equation}
where $\zeta(m(t))$ is a random variable by which:
\[
p(\zeta(m(t))=0)=\frac{m(t)}{m(t+1)}\;,\, p(\zeta(m(t))=1)=\frac{1}{m(t+1)}\,.
\]
Clearly,
\[
m(t)=1+\sum_{t'\leq t} \zeta(m(t))\,,
\]
is a non decreasing function in $t$.  Now let us define $M_k(t)$ as a random variable taking values uniformly at random over the set $m_1, . . .,m_k, m_{k+1}$ at time $t$. The sequence of random variables accounting describing the CRPM $\eta(t)$ can be written as:
\[
%\eta(t)=M_1(1), M_1(2), M_{m(2)}(3), M_{m(3)}(4), . . .,M_{m(t-1)}(t), M_{m(t)}(t+1), . . .\,,
\eta(t)=M_1(1), M_1(2), M_{m(2)}(3), . . ., M_{m(t)}(t+1), . . .\,,
\]
leading to paths $x(t)$ of the kind:
\[
x(t)=m_1, m_1, m_2, m_1, m_2, m_3, m_2, m_2, m_1, m_3, . . .
\]
We emphasize that {\em the tables visited are distinguishable and can visited repeatedly}. 
The CRPM however does not fill up $\Omega(t)$, i.e. there exist elements in $\Omega(t)=\times_{t'=1}^t\Omega_{t'}$, that are not potential paths of the CRP. This includes all sequences that select a table $m_i$ without ever having chosen some table $m_j$, with $j<i$ before. For example, the path $x(t)=(m_1,m_2,m_1,m_2,m_3,m_1,m_5,m_4,m_3,\cdots)$ is not possible, because $m_5$ is chosen before $m_4$. The CRPM therefore gives us the opportunity to introduce sampling spaces conditional to a particular well formed path. In this particular case, it is enough to observe that the sampling space some well formed path $x(t)\in\Omega(t-1)$ {\em sees} at time $t$ is given by:
\begin{equation}
\label{eq:condstatesamplespace}
\Omega_t[x(t-1)]=\{m_1,m_2,\cdots,m_{m[x](t-1)]+1} \}\,,
\end{equation} 
where $m[x](t-1)$ is the number of different tables the well formed CRPM path $x(t-1)$ has sampled at time $t-1$. By convention, we define $m[x](1)=1$. 

\subsection{Statistics of the CRPM}
$\\$

We start computing the probability of a particular path $x(t)\in\mathring{\Omega}(t)$.
\begin{proposition}
\label{Sprop:pxt}
Let process $\eta(t)$ be the CRPM and $x(t)\in \mathring{\Omega}(t)$ be a given path of the process up to time $t$. Let $m[x](t')$ be the number of occupied tables in the restaurant 
associated with the path $x(t')$ at time $t'$. Then the probability to observe the particular sequence of tables, $p(x(t))$ is given by:
\[
p(x(t))=\prod_{t'=1}^{t}\left(m[x](t')+1\right)^{-1}\,.
\]
\end{proposition}
\begin{proof}
The proposition follows from direct calculation.
\end{proof}
Now we will see that the sequence corresponding to the number of occupied tables converges to a tractable functional form.

\begin{proposition}
\label{prop:tosqrtmu}
Given the sequence of occupied tables of the CRPM $m(1), . . .,m(t)$ as defined in equation (\ref{Seq:m(t)}), then:
\[
\frac{m(t)}{\sqrt{2t}}\to 1\,,
\]
in probability.
\end{proposition}

\begin{proof}
Consider the random variable $\Delta t_k$ denoting the amount of steps by which $m(t)=k$. $\Delta t_k$ is a geometric random variable with associated law:
\[
p(\Delta t_k=i)=\left(1-\frac{1}{k}\right)^{i-1}\frac{1}{k}\;,\, \langle \Delta t_k\rangle=k\;,\, \sigma^2(\Delta t_k)=k^2-k\,.
\]
Now we construct a new set of renormalized random variables, $\delta t_1, . . .,\delta t_k$ as:
\[
\delta t_k\equiv\frac{\Delta t_k}{k}\;,\,\langle \delta t_k\rangle=1\;,\,\sigma^2(\delta t_k)=1-\frac{1}{k}\,. 
\]
In that context, the sum of $\delta t_1, . . .,\delta t_k$ is the sum of $k$ random variables with mean $1$ and $\sigma^2<1$. Therefore, there exists a monotonously increasing function, $\mu(t)$ by which, by the law of large numbers, for each pair $\epsilon,\delta<0$, there exists $t_0$ such that, for $t>t_0$:
\[
\mathbb{P}\left(\left|\frac{\sum_{k\leq m(t)}\delta t_k}{\mu(t)}-1\right|>\epsilon\right)<\delta\,.
\]
Notice that, in this setting, we have that the deviations behave close to a discrete random walk centered at $0$ and with step length $\sigma^2<1$.
Since for all $\delta t_k$, $\langle \delta t_k\rangle=1$:
\[
\frac{m(t)}{\mu(t)}\to 1\;,\;{\rm (in\; probability)}\,,
\]
and, since, by construction:
\[
\sum_{k\leq \mu(t)}k \delta t_k =t\,,
\]
one has that:
\[
\mu(t)\sim \sqrt{2t}\,.
\]
where "$\sim$" means asymptotically equivalent, as we wanted to demonstrate.
\end{proof}

\subsection{The typical set of the CRP}
\begin{theorem}
The CRPM is compact with compact scale $(\Lambda, g(t))=(\log, \frac{t}{2}\log t)$.
\label{Sth:compactCrp}
\end{theorem}

\begin{proof}
We need to demonstrate that:
\[
\frac{1}{\frac{t}{2}\log t}\log\left(\frac{1}{p(\eta(t))}\right)\to 1\,,\;\;{\rm (in\; probability)}\,.
\]
We first note that, according to the definition of the sequence of the number of occupied tables given in equation (\ref{Seq:m(t)}), and the statement of proposition \ref{Sprop:pxt}, one can rewrite the logarithmic term of the condition for compactness as:
\[
\log\left(\frac{1}{p(\eta(t))}\right)=\sum_{t'\leq t}\log(m(t'))\,.
\]
Now, let us define a new random variable $z(t)$ as follows:
\[
z(t)=m(t)-\mu(t)\,,
\]
with $\mu(t)=\sqrt{2t}$. Notice that, according to proposition \ref{prop:tosqrtmu}, we have that:
\[
\frac{z(t)}{\mu(t)}\to 0\;,\;{\rm (in\; probability)}\,.
\]
We can then rewrite condition of compactness, with $g(t)=\sum_{t'\leq t}\log(\mu(t))$, as:
\begin{widetext}
\begin{eqnarray}
\frac{1}{\sum_{t'\leq t}\log(\mu(t'))}\log\left(\frac{1}{p(\eta(t))}\right)&=&\frac{1}{\sum_{t'\leq t}\log(\mu(t'))}\sum_{t'\leq t}\log(m(t'))\nonumber\\
&=& \frac{1}{\sum_{t'\leq t}\log(\mu(t'))}\sum_{t'\leq t}\log(\mu(t')+z(t'))\nonumber\\
&=& \frac{1}{\sum_{t'\leq t}\log(\mu(t'))}\sum_{t'\leq t}\log(\mu(t'))+\log\left(1+\frac{z(t')}{\mu(t')}\right)\nonumber\\
&=&1+\frac{1}{\sum_{t'\leq t}\log(\mu(t'))}\sum_{t'\leq t}\log\left(1+\frac{z(t')}{\mu(t')}\right)\,.\nonumber
\end{eqnarray}
\end{widetext}
It remains to see that the second term of the sum goes to $0$. Clearly, by proposition \ref{prop:tosqrtmu}:
\[
\frac{\log \left(1+\frac{z(t)}{\mu(t)}\right)}{\log(\mu(t))}\to 0\;,\;\;{\rm (in\; probability)}\,.
\]
Consequently:
\[
\frac{1}{\sum_{t'\leq t}\log(\mu(t'))}\log\left(\frac{1}{p(\eta(t))}\right)\to 1\;,\;\;{\rm (in\; probability)}\,.
\]
Finally, we need to compute the asymptotic form of  $\sum_{t'\leq t}\log(\mu(t'))$. Observing that we have a Riemann sum, one can consider:
\[
\sum_{t'\leq t}\log(\mu(t'))\sim\int^t \log(\sqrt{2t'})dt' \sim \frac{t}{2}\log t\,,
\]
thus concluding the proof.
\end{proof}

\subsection{The entropy of the CRP}
$\\$

We demonstrate here that the CRPM is simple. In consequence, the typical set can be computed as a function of the entropy. In that case, one can show that a suitable choice is $\Lambda=\log$ --chosen by the sake of simplicity--, implying that the associated entropy is Shannon path entropy. However, we emphasize that this choice is not unique. Given a different choice of $g$, one could have another $\Lambda$ by which the process is also simple and, therefore, the typical set could be defined through another form of entropy. We briefly comment this point in the next section, sketching how another potential pair $(\Lambda, g)$ functions would work as well.

\begin{theorem}
The CRPM is simple in $(\Lambda, g(t))=(\log, \frac{t}{2}\log t)$.
\label{Sth:CRPSimple}
\end{theorem}

\begin{proof}
Since the CRPM is compact with compact scale $(\Lambda, g(t))=(\log, \frac{t}{2}\log t)$, we know that there is a typical localizer sequence $\epsilon_1,. . .,\epsilon_t, . . .$ with associated $\delta_1, . . .,\delta_t, . . .$, such that $\lim_{t\to\infty}\epsilon_t=\lim_{t\to\infty}\delta_t=0$. The paths belonging to the typical set $A[\epsilon_t]$, are those satisfying:
\[
e^{-(1+\epsilon_t)\frac{t}{2}\log t}\leq p(x(t))\leq e^{-(1-\epsilon_t)\frac{t}{2}\log t}\,.
\]
In addition, the measure associated to the typical set $A[\epsilon_t]$ is given by:
\[
\mathbb{P}(A[\epsilon_t])\geq 1-\delta_t\,.
\]
In consequence, 
\begin{eqnarray}
(1-\epsilon_t)(1-\delta_t)\frac{t}{2}\log t&\leq& -\sum_{x(t)\in A[\epsilon_t]} p(x(t))\log p(x(t))\nonumber\\
&\leq& (1+\epsilon_t)(1-\delta_t)\frac{t}{2}\log t\,.\nonumber
\end{eqnarray}
Now let's consider that the complement of the typical set ${\Omega}(t)\setminus A[\epsilon_t]$, whose measure is:
\[
\mathbb{P}({\Omega}\setminus A[\epsilon_t])\leq\delta_t\,,
\]
is completely populated by those paths by which:
\[
x(t)\in {\Omega}\setminus A[\epsilon_t]\;,\, p(x^*(t))=\min_{\mathring{\Omega}(t)}\{p(x(t))\}\,,
\]
therefore, for all paths $x(t)\in \mathring{\Omega}(t)$:
\[
-\log p(x(t))\leq -\log p^*(x(t))\,.
\]
The least probable path is the one that increases the number of tables at every step, having it probability $p^*(x(t))$ of:
\[
p^*(x(t))=\frac{1}{t!}\,,
\]
leading to $-\log p^*(x(t))\sim t\log t$, thanks to the Stirling's approximation \cite{Abramowitz:1964}. In that context:
\[
-\sum_{x(t)\in {\Omega}(t)\setminus A[\epsilon_t]} p(x(t))\log p(x(t))\leq \delta t\log t\,.
\]
Collecting the above reasoning, and by observing that the defined entropy is actually the {\em Shannon path entropy}, $S_\Lambda(\eta(t))=H(\eta(t))$ \cite{Cover:2012}:
\[
H(\eta(t))=-\sum_{x(t)\in {\Omega}(t)} p(x(t))\log p(x(t))\,,
\]
one has that:
\begin{eqnarray}
(1-\epsilon_t)(1-\delta_t)\frac{t}{2}\log t &\leq& H(\eta(t))\nonumber\\
&\leq & (1+\epsilon_t)(1-\delta_t)\frac{t}{2}\log t+ \delta_t t\log t\,.\nonumber
\end{eqnarray}
In consequence, by defining $g(t)=\frac{t}{2}\log t$, one is led to:
\[
(1-\delta_t)(1-\epsilon_t)\leq \frac{H(\eta(t))}{g(t)}\leq (1-\delta_t)(1+\epsilon_t)+2\delta_t\,.
\]
Since we know that $\lim_{t\to\infty}\epsilon_t=\lim_{t\to\infty}\delta_t=0$ we conclude that:
\[
\frac{H(\eta(t))}{g(t)}\to 1\;,
\]
as we wanted to demonstrate.
\end{proof}

\subsection{A different compact scale $(\Lambda, g)$ for the CRP}
\label{secSI:Lg}
$\\$

We finally briefly comment how another compact scale made of different $(\Lambda, g)$ functions could be used to characterize the typical set and the generalized entropy of the CRP.  We avoid the technical details, for the sake of simplicity. We start computing the inverse function of $g$, $g^{-1}$:
\[
g^{-1}(z)\approx \frac{2z}{\mathbf{W}(2z)}\,,
\]
where $\mathbf{W}$ is the {\em Lambert} function \cite{Abramowitz:1964}, where only the positive, real branch is taken into account. 
Then we compose it with the $\log$ function, thereby defining a new function $\Lambda$ as:
\begin{equation}
\Lambda(z)=(g^{-1}\circ\log)(z)\sim \frac{2\log(z)}{\mathbf{W}(2\log(z))}\,.
\label{Seq:Lambdacirc}
\end{equation}
We observe that $\Lambda$ as above defined is a strictly growing, concave function with continuous second derivatives.
Clearly, $(g^{-1}\circ g)(t)\sim t$. Therefore, as a direct consequence of theorem \ref{Sth:compactCrp}, the CRPM is compact in $(\Lambda, t)$:
\[
\frac{1}{g^{-1}\left(\frac{t}{2}\log t\right)}(g^{-1}\circ \log)\left(\frac{1}{p(\eta(t))}\right)=\frac{1}{t}\Lambda\left(\frac{1}{p(\eta(t))}\right)\to 1\;.
\]
(in probability). In consequence, the size of the typical set and the typical probabilities can be approximated from the following generalized entropy $S_\Lambda$:
\begin{equation}
S_{\Lambda}(\eta(t))=\sum_{x(t)\in{\Omega}(t)}p(x(t))\frac{2\log\left(\frac{1}{p(x(t))}\right)}{\mathbf{W}\left[2\log\left(\frac{1}{p(x(t))}\right)\right]}\;.
\label{eq:ExtCh}
\end{equation}
where we explicitly wrote it in terms the functional form of $\Lambda=g^{-1}\circ\log$ as defined in equation (\ref{Seq:Lambdacirc}).

\end{document}